\documentclass[letterpaper,11pt,onecolumn]{article}

\usepackage[]{epsf,epsfig, amsmath,amssymb,amsthm, latexsym, color,graphicx, xspace, url}

\usepackage{times,euscript, balance}

\usepackage{xspace}

\newtheorem{theorem}{Theorem}[section]

\newtheorem{claim}[theorem]{Claim}
\newtheorem{proposition}[theorem]{Proposition}

\newtheorem{corollary}[theorem]{Corollary}
\newtheorem{definition}[theorem]{Definition}

\def\eps{{\varepsilon}}

\def\A{\EuScript{A}}

\def\B{\EuScript{B}}

\def\D{\EuScript{D}}

\def\F{\EuScript{F}}

\def\M{\EuScript{M}}
\def\P{\EuScript{P}}

\def\S{\EuScript{S}}

\def\etal{\textit{et~al.}}

\def\disc{{\rm{disc}}}
\def\reals{{\mathbb R}}

\newdimen\instindent
\def\institute#1{\gdef\@institute{#1}}

 \newfont{\affaddr}{phvr at 11pt}
 \newfont{\affaddrit}{phvro at 11pt} 

\usepackage{graphicx,floatflt,psfrag}

\begin{document}

\begin{titlepage}
\title{A Size-Sensitive Discrepancy Bound for Set Systems of Bounded Primal Shatter Dimension\thanks{
    Work on this paper has been supported by NSF under grant CCF-12-16689.
  }}

\author{Esther Ezra\thanks{%
    Courant Institute of Mathematical Sciences,
    New York University, New York, NY 10012, USA;
    \textsl{esther@courant.nyu.edu}.
  } 
}

\maketitle

\begin{abstract}
  Let $(X,\S)$ be a set system on an $n$-point set $X$.
  The \emph{discrepancy} of $\S$ is defined as the minimum of the largest deviation 
  from an even split, over all subsets of $S \in \S$ and two-colorings $\chi$ on $X$.
  We consider the scenario where, for any subset 
  $X' \subseteq X$ of size $m \le n$ and for any parameter $1 \le k \le m$, the number of restrictions of the sets
  of $\S$ to $X'$ of size at most $k$ is only $O(m^{d_1} k^{d-d_1})$, for fixed integers $d > 0$ and $1 \le d_1 \le d$
  (this generalizes the standard notion of \emph{bounded primal shatter dimension} when $d_1 = d$).
  In this case we show that there exists a coloring $\chi$ with discrepancy bound 
  $O^{*}(|S|^{1/2 - d_1/(2d)} n^{(d_1 - 1)/(2d)})$,
  for each $S \in \S$, where $O^{*}(\cdot)$ hides a polylogarithmic factor in $n$. 
  This bound is tight up to a polylogarithmic factor~\cite{Mat-95, Mat-99} and the corresponding coloring $\chi$ can be computed in 
  expected polynomial time using the very recent machinery of Lovett and Meka for constructive discrepancy minimization~\cite{LM-12}. 
  Our bound improves and generalizes the bounds obtained from the machinery of Har-Peled and Sharir~\cite{HS-11} (and the follow-up work in~\cite{SZ-12}) 
  for points and halfspaces in $d$-space for $d \ge 3$. 
\end{abstract}
\end{titlepage}


\vspace{-2ex}
\section{Introduction}
\label{sec:intro}
\vspace{-2ex}


Let $(X, \S)$ be a finite set system 
with $n = |X|$. 
A \emph{two-coloring} of $X$ is a mapping $\chi : X \rightarrow \{-1 ,+1\}$.
For a subset $S \in \S$ we define $\chi(S) := |\sum_{x \in S} \chi(x)|$.
The \emph{discrepancy} of $\S$ is then defined as
$$
\disc(\S) := \min_{\chi} \max_{S \in \S} \chi(S)  .
$$
In other words, the discrepancy of the set system $(X,\S)$ is the minimum over all colorings $\chi$ of the 
largest deviation from an even split, over all subsets in $\S$.

Our goal in this paper is to derive discrepancy bounds for $(X,\S)$ in the scenario where $(X,\S)$ admits 
a polynomially bounded \emph{primal shatter function} and has some additional favorable properties.
In the bounds that we derive the discrepancy for each $S \in \S$ is sensitive to its cardinality $|S|$.
Let us first recall the definition of set systems of this kind:

\begin{definition}[Primal Shatter Function; Matou\v{s}ek~\cite{Mat-99}]
  The \emph{primal shatter function} of a set system $(X,\S)$ is a function, denoted by $\pi_{\S}$, whose value
  at $m$ is defined by
  $$
  \pi_{\S}(m) = \max_{Y \subseteq X, |Y| = m} |\S |_{Y}| ,
  $$
  where $\S |_{Y}$ is the collection of all sets in $\S$ projected onto (that is, restricted to) $Y$. 
  In other words, $\pi_{\S}(m)$ is the maximum possible number of distinct intersections of the sets of $\S$
  with an $m$-point subset of $X$.
\end{definition}

From now on we say that a set system $(X,\S)$ with $|X| = n$ (where $n$ can be assumed to be arbitrarily large) 
has a \emph{primal shatter dimension $d$} if $\pi_{\S}(m) \le C m^d$, for all $m \le n$, where $d > 1$ and $C > 0$ are constants.

A typical family of set systems that arise in geometry with bounded primal shatter dimension consists of set systems $(X,\S)$
of points in some low-dimensional space ${\reals}^d$, and $\S$ is a collection of certain simply-shaped regions, e.g., 
halfspaces, balls, or simplices (where $d > 0$ is assumed to be a constant). In such cases, the primal shatter function is $m^{O(d)}$; 
see, e.g.,~\cite{Har-Peled-11} for more details. In fact, set systems of this kind are part of a more general family, 
referred to as set systems of \emph{finite VC-dimension}~\cite{VC-71}; the reader is referred to~\cite{Har-Peled-11, HW87} for the exact definition.
Although the ``official'' definition of finite VC-dimension is different, it suffices to have the same requirement as for set systems of 
polynomially bounded primal shatter function. It is also known 
that the VC-dimension is finite if and only if the primal shatter dimension is finite, although they do not necessarily have the same 
value, see, e.g.,~\cite{Har-Peled-11} for more details. From now on we make only the assumption about having a finite primal shatter 
dimension, in particular, this is the case in our construction and analysis, and the VC-dimension is mentioned here only for the sake 
of completeness of the presentation. 


A major result by Matou\v{s}ek~\cite{Mat-95} (see also~\cite{Mat-99, MWW-93}) is the following:

\begin{theorem}[Matou\v{s}ek~\cite{Mat-99}]
  \label{the:primal_shatter}
  Let $(X,\S)$ be a set system as above with $|X| = n$, $\pi_{\S}(m) \le C m^d$, for all $m \le n$, where $d > 1$ and $C > 0$ 
  are constants. Then 
  $$
  \disc(\S) = O(n^{1/2 - 1/(2d)}) ,
  $$ 
  where the constant of proportionality depends on $d$ and $C$.
\end{theorem}

This bound is known to be tight in the worst case (see~\cite{Mat-99} and the references therein for more details).

\vspace{-2ex}
\paragraph{Relative $(\eps, \delta)$-approximations and $\eps$-nets.}

The motivation to establish sensitive discrepancy bounds of the above kind stems from their application in the construction of
\emph{relative $(\eps,\delta)$-approximations}, introduced by Har-Peled and Sharir~\cite{HS-11}\footnote{In~\cite{HS-11} they were
  introduced, with a slighly different notation, as relative $(p,\eps)$-approximations.}
based on the work of Li~\etal~\cite{LLS-01} in the context of machine learning theory. We recall the definition from~\cite{HS-11}:
For a set system $(X,\S)$ (with $X$ finite), the \emph{measure} of a set $S \in \S$ is the quantity
$\overline{X}(S) = \frac{|S \cap X|}{|X|}$.
Given a set system $(X ,\S)$ and two parameters, $0 < \eps <1$ and $0 < \delta < 1$,
we say that a subset $Z \subseteq X$ is a \emph{relative $(\eps,\delta)$-approximation} if it satisfies,
for each set $S \in \S$,
$$
\overline{X}(S)(1-\delta) \le \overline{Z}(S) \le \overline{X}(S)(1+\delta) ,
\quad \mbox{if $\overline{X}(S) \ge \eps$,} \quad \mbox{and}
$$
$$
\overline{X}(S) -\delta \eps \le \overline{Z}(S) \le \overline{X}(S) + \delta \eps , \quad \mbox{otherwise.}
$$
A strongly related notion is the so-called \emph{$(\nu,\alpha)$-sample}~\cite{Har-Peled-11, Haussler-92, LLS-01}),
in which case the subset $Z  \subseteq X$ satisfies, for each set $S \in \S$,
$$
d_{\nu}(\overline{X}(S), \overline{Z}(S)) := \frac{|\overline{Z}(S) - \overline{X}(S)|}{\overline{Z}(S) + \overline{X}(S) + \nu} < \alpha .
$$
As observed by Har-Peled and Sharir~\cite{HS-11}, relative $(\eps,\delta)$-approximations and $(\nu,\alpha)$-samples
are equivalent with an appropriate relation between $\eps$, $\delta$, and $\nu$, $\alpha$ (roughly speaking, they are equivalent
up to some constant factor). Due to this observation they conclude that the analysis of Li~\etal~\cite{LLS-01} (that shows a bound
on the size of $(\nu,\alpha)$-samples) implies that
for set systems of finite VC-dimension $d$, there exist relative $(\eps,\delta)$-approximations of size 
$\frac{c d\log{(1/\eps)}}{\delta^2\eps}$, where $c > 0$ is an absolute constant. In fact, any random sample of these many elements of
$X$ is a relative $(\eps,\delta)$-approximation with constant probability. 
More specifically, success with probability at least $1-q$ is guaranteed if one samples
$\frac{c (d\log{(1/\eps)} + \log{(1/q)})}{\delta^2 \eps}$ elements of $X$.\footnote{We note that although in the original analysis for this bound 
  $d$ is the VC-dimension, this assumption can be replaced by having just a primal shatter dimension $d$; 
  see, e.g.,~\cite{Har-Peled-11} for the details of the analysis.}

It was also observed in~\cite{HS-11} that \emph{$\eps$-nets} 
arise as a special case of relative $(\eps,\delta)$-approximations.
Specifically, an $\eps$-net is a subset $N\subseteq X$ with the property that any set $S \in \S$ with $|S \cap X|\ge \eps|X|$
contains an element of $N$; in other words, $N$ is a hitting set for all the ``heavy'' sets.
In this case, if we set $\delta$ to be some constant fraction, say, $1/4$, then a relative $(\eps,1/4)$-approximation becomes
an $\eps$-net. Moreover, a random sample of $X$ of size $O\left(\frac{\log{(d/\eps)} + \log{(1/q)}}{\eps}\right)$, with an appropriate 
choice of the constant of proportionality, is an $\eps$-net with probability at least $1-q$; see~\cite{HS-11} for further details.
Our analysis exploits these two structures and the relation between them---see below.

\vspace{-2ex}
\paragraph{Related work.}

There is a rich body of literature in discrepancy theory, with numerous bounds and results. It is beyond the scope of this paper to mention all results, 
and we just list those that are most relevant to our work.
We refer the reader to the book of Chazelle~\cite{Chaz-01} for an overview of discrepancy theory and the book of Matou{\v s}ek~\cite{Mat-99} 
for various results in geometric discrepancy. In particular, our work is based on the techniques overviewed in the latter.

We first briefly overview previous results for an abstract set system on an $n$-point set $X$, with $m = |\S|$ sets. 
The celebrated ``six standard deviations'' result of Spencer~\cite{Spencer-85},
which is an extension to the partial coloring method of Beck~\cite{Beck-81}, implies that for such set systems there exists a coloring $\chi$ such that
$\chi(S) \le K \sqrt{n (1 + \log{(m/n)})}$, for each $S \in \S$, where $K > 0$ is a universal constant. In particular, when $m = n$ we have $K = 6$, 
in which case the discrepancy bound becomes $6 \sqrt{n}$. In contrast, one can easily show that using simple probabilistic considerations, a random 
coloring yields a (suboptimal) discrepancy bound of $\sqrt{n \log{m}}$ (or $\sqrt{n \log{n}}$ if $m = O(n)$). 
A long-standing open problem was whether the result of Spencer~\cite{Spencer-85} can be made constructive, and this has recently been answered 
in the positive by Bansal~\cite{Bansal-10} for the case $m = O(n)$ (for the general case his bound is slightly suboptimal with respect to 
the bound in~\cite{Spencer-85}). 
In a follow-up work, Lovett and Meka~\cite{LM-12} have shown a new elementary constructive proof of Spencer's result, resulting in the same asymptotic 
discrepancy bounds as in~\cite{Spencer-85}, for arbitrary values of $m$. 

In geometric set systems, upper bounds were first shown by Beck, where the Lebesgue measure of a class of geometric shapes is 
approximated by a discrete point set; see the book by Beck and Chen~\cite{BC-87}. 
For arbitrary points sets, Matou\v{s}ek~\etal~\cite{MWW-93} have addressed the 
case of points and halfspaces in $d$-space, and showed an almost tight bound of $O(n^{1/2 - 1/(2d)} \sqrt{\log{n}})$, which has later been
improved to $O(n^{1/2 - 1/(2d)})$~\cite{Mat-95, Mat-99} (Theorem~\ref{the:primal_shatter}).
Concerning lower bounds, there is a rich literature where several such bounds are obtained in geometric set systems. We only mention the lower bound 
$\Omega(n^{1/2 - 1/(2d)})$ of Alexander~\cite{Alexander-90} for set systems of points and halfspaces in $d$-space. For further results we refer the
reader to~\cite{Mat-95} and the references therein.

The extension of discrepancy bounds for points and halfspaces in $d$-space to be size-sensitive has been addressed in the work of Har-Peled 
and Sharir~\cite{HS-11}, who showed a bound of $O(|S|^{1/4} \log{n})$ in the two-dimensional case, using an intricate extension of the technique of 
Welzl~\cite{Welzl-92} (see also~\cite{CW-89}) for constructing spanning trees with low crossing numbers. 
Nevertheless, their technique cannot be applied in higher dimensions, 
because already at $d=3$ they showed a counterexample to their construction. The follow-up work of Sharir and Zaban~\cite{SZ-12} (based on the construction
in~\cite{HS-11}) addresses these cases, establishing the bound $O\left(n^{(d - 2)/2(d - 1)} |S|^{1/(2d(d - 1))} \log^{(d+1)/(2(d - 1))}{n}\right)$. 

\vspace{-2ex}
\paragraph{Our results.}

In this paper we refine the bound in Theorem~\ref{the:primal_shatter} so that it becomes sensitive to the size of the 
sets $S \in \S$ in several favorable cases.
Specifically, we assume that for any (finite) set system projected onto a ground set of 
size $m \le n$ and for any parameter $1 \le k \le m$, the number of sets
of size at most $k$ is only $O(m^{d_1} k^{d-d_1})$, where $d$ is the primal shatter dimension and
$1 \le d_1 \le d$ is an integer parameter\footnote{We ignore the cases where $d_1=0$ or $d_1$ takes fractional values,
  as they do not seem to appear in natural set systems, and, in particular, in the geometric set systems that we consider.}. 
By assumption, when $k = m$ we obtain $O(m^{d})$ sets in total, in accordance with the assumption that the primal shatter dimension is 
$d$, but the above bound is also sensitive to the size $k$ of the set.
 

We show that for set systems of this kind there exists a coloring $\chi$ such that 
$$
\chi(S) = O^{*}(|S|^{1/2 - d_1/(2d)} n^{(d_1 - 1)/(2d)}) , 
$$
where $O^{*}(\cdot)$ hides a polylogarithmic factor in $n$. 
This bound almost matches (up to the polylogarithmic factor) the optimal discrepancy bound in 
Theorem~\ref{the:primal_shatter} when $|S| = n$,
but is more general than this bound as it yields a whole range of bounds for $1 \le d_1 \le d$.
Specifically, when $d_1 = d$, the number of sets is just $O(m^{d})$ (that is, it is not sensitive to their size)
in which case we just have a set system of bounded primal shatter dimension, and then, once again, our discrepancy bound
almost matches the optimal bound in Theorem~\ref{the:primal_shatter}.  
In the other extreme scenario, when $d_1 = 1$, the set system is what we call ``well-behaved'', and our discrepancy
bound then becomes $O^{*}(|S|^{1/2 - 1/2d})$, which depends only on $|S|$ (up to the polylogarithmic factor).
In particular, set systems of points and halfspaces in $d$ dimensions are well-behaved with $d=2$ and $d=3$, respectively.
In the plane, the resulting bound is slightly suboptimal with respect to the sensitive bound of Har-Peled and Sharir~\cite{HS-11}
(our hidden polylogarithmic factor is slightly larger than their $\log{n}$ factor).
For points and halfspaces in three and higher dimensions, our bound considerably improves the bound in~\cite{SZ-12} 
(which extends the construction in~\cite{HS-11}). In particular, the bound in the three-dimensional case in~\cite{SZ-12} 
is not purely sensitive in $|S|$ but also contains an additional sublinear term in $n$ (whereas the original technique
of Har-Peled and Sharir~\cite{HS-11} failed to yield such a bound).


Our construction uses a different machinery than the one in~\cite{HS-11, SZ-12} 
and is a variant of the construction of Matou\v{s}ek~\cite{Mat-99} (see also Matou\v{s}ek~\cite{Mat-95} 
for the special case of points and halfspaces in $d$-space),
based on \emph{Beck's partial coloring} and on the \emph{entropy method}, and is combined with
the property that set systems with bounded primal shatter function admit a small \emph{packing} (see below for details concerning these notions).
Our assumption about the set system (that is, the bound on the number of sets of size at most $k \le m \le n$ stated above)  
enables us to refine the analysis in~\cite{Mat-99} to be size sensitive, which eventually leads to the desired bound.

\vspace{-2ex}
\section{Preliminaries}
\label{sec:prelim}
\vspace{-2ex}

We now briefly review some of the tools mentioned in Section~\ref{sec:intro}, which are applied by our analysis.
Then, we present the construction. 

\vspace{-2ex}
\paragraph{Partial coloring and entropy.}

Let $X$ be a set as above. 
A \emph{partial coloring} of $X$ is any mapping $\chi:X \rightarrow \{ -1, 0, +1 \}$.
A point $x \in X$ with $\chi(x) = 1$ or $\chi(x) = -1$ is said to be 
\emph{colored by $\chi$}, whereas $\chi(x) = 0$ means that $x$ is \emph{uncolored}.

A method originated by Beck~\cite{Beck-81}, which was subsequently elaborated by others and is known by now
as the \emph{entropy method} (see, e.g.,~\cite{AS-00, MS-96, Spencer-85}), shows that under some favorable assumptions 
there exists a substantial partial coloring of $(X,\S)$ with good discrepancy. 
In our analysis we apply the entropy method as a black box, and thus only
present the constructive version of Beck's partial coloring lemma as has very recently been formulated by Lovett and Meka~\cite{LM-12};
see Section~\ref{sec:intro} and the references therein for more details.
We also refer the reader to~\cite{Mat-98} for the original (non-constructive) version of the entropy method.

%

\begin{proposition}[Lovett and Meka~\cite{LM-12}]
  \label{prop:constructive_disc}
  Let $(X,\S)$ be a set system with $|X| = n$, and let $\Delta_S > 0$ be some real number assigned to $S$,
  for each $S \in \S$.
  Suppose that 
  $$
  \sum_{S \in \S} \exp{\left( -\frac{\Delta_S^2}{16 |S|} \right)} \le \frac{n}{16} .
  $$
  Then there exists a partial coloring $\chi:X \rightarrow \{ -1, 0, +1 \}$ that colors at least $n/2$ points of $X$,
  so that $\chi(S) \le \Delta_S + 1/n^c$, for each $S \in \S$, where $c > 0$ is an arbitrarily large constant.
  Moreover, $\chi$ can be computed in expected polynomial time in $n$ and $|\S|$ (where the degree of the polynomial bound depends on $c$).
\end{proposition}
  
%

\vspace{-2ex}
\paragraph{$\delta$-packing.}

Let $(X,\S)$ be a set system as above, and let $\delta > 0$ be a given integer parameter.
A \emph{$\delta$-packing} is a maximal subset $\P \subseteq \S$ such that the symmetric difference distance 
$|S_1 \triangle S_2|$ between any pair of sets $S_1, S_2 \in \P$ is strictly greater than $\delta$.
We also call such a set $\P$ \emph{$\delta$-separated}.

A key property, shown by Haussler~\cite{Haussler-95} (see also~\cite{Chazelle-92, CW-89, Mat-99, Welzl-88}), is that 
set systems with a bounded primal shatter dimension admit small $\delta$-packings. That is:

\begin{theorem}[Packing Lemma~\cite{Chazelle-92, Haussler-95}]
  \label{thm:packing}
  Let $d > 1$ be a constant, and let $(X,\S)$ be a set system on an $n$-point set with
  primal shatter dimension $d$.
  Let $1 \le \delta \le n$ be an integer parameter, and let $\P \subseteq \S$ be $\delta$-separated. 
  Then $|\P| = O((n/\delta)^{d})$.
\end{theorem}


\vspace{-2ex}
\section{The Construction}
\label{sec:construction}
\vspace{-2ex}

We are now ready to present our construction.
Let $(X,\S)$ be a set system of bounded primal shatter dimension $d$, with the additional property that the number of
sets of size at most $k$ in the projection of $(X,\S)$ onto any $m$-point subset $X' \subseteq X$ 
is $O(m^{d_1} k^{d-d_1})$, for any $m \le n$.
We construct a decomposition of each set $S \in \S$ as a Boolean combination of ``canonical sets''
obtained from a sequence of packings that we build for $\S$. 
This decomposition is a variant of the one presented in~\cite[Chapter 5]{Mat-99} (see also~\cite{Mat-95}),
and is also referred to as \emph{chaining} in the literature (see, e.g.,~\cite{Har-Peled-11}).

\paragraph{Chaining.}

For the sake of completeness, we repeat some of the details in~\cite[Chapter 5]{Mat-99} and use similar notation
for ease of presentation.
Without loss of generality, we assume that $n = 2^{k}$, for some integer $k > 0$. 
For $j = 0, \ldots, k$, let $\F_j \subseteq \S$ be an $(n/2^{j})$-packing that is maximal under inclusion.
By definition, this implies that for any pair of sets in $\F_j$, their 
distance is larger than $n/2^{j}$ and this set is maximal with respect to this property. 
In particular, we have $\F_0 = \{\emptyset\}$, and $\F_k = \S$ by construction. 

Observe that for each $F_j \in \F_j$, there is a set $F_{j-1} \in \F_{j-1}$ with $|F_j \triangle F_{j-1}| \le n/2^{j-1}$.
This follows from the inclusion-maximality of $\F_{j-1}$.
That is, consider the set $F_{j-1}$ \emph{closest} to $F_j$ in $\F_{j-1}$.
Either $F_j = F_{j-1}$ and then the claim is trivial, or $F_j \neq F_{j-1}$, and then the opposite inequality 
$|F_j \triangle F_{j-1}| > n/2^{j-1}$ is impossible, for then $F_j$ could be added to $\F_{j-1}$, contradicting its maximality.
Let us attach to each $F_j \in \F_j$ its nearest neighbor (closest set) $F_{j-1} \in \F_{j-1}$, 
and put $A(F_j) := F_j \setminus F_{j-1}$, $B(F_j) := F_{j-1} \setminus F_j$.  We now form the set systems
$\A_j := \{ A(F_j) \mid F_j \in \F_j \}$, $\B_j := \{ B(F_j) \mid F_j \in \F_j \}$, $j = 1, \ldots, k$.
It has been observed in~\cite[Chapter 5]{Mat-99} that each set $S \in \S$ can be decomposed as
\begin{equation}
\label{eq:decomposition1}
  S = (\ldots (((A_1 \setminus B_1) \sqcup A_2) \setminus B_2) \sqcup \cdots \sqcup A_k) \setminus B_k ,
\end{equation}
where $\sqcup$ denotes disjoint union, and $A_j \in \A_j$, $B_j \in \B_j$, for each $j = 1, \ldots, k$.
Indeed, consider the nearest-neighbor ``chain'' 
$S = F_k \rightarrow F_{k-1} \rightarrow \cdots \rightarrow F_{0} = \emptyset$ (recall that $F_0 = \emptyset$ by
our assumption on $\F_0$).
Each set $F_j \in \F_j$ on this chain can be turned into its nearest neighbor $F_{j-1} \in F_{j-1}$ by adding
$B(F_j)$ and subtracting $A(F_j)$. Moreover, it is easy to verify using induction on $j \ge 1$ that 
$F_j = (\ldots (((A_1 \setminus B_1) \sqcup A_2) \setminus B_2) \sqcup \cdots \sqcup A_j) \setminus B_j$, 
and $S$ is obtained at $j=k$ (see also~\cite{MN-12} for similar considerations).


We next refine this decomposition as follows.
We partition the sets in $\S$ into $k+1$ subfamilies $\S_0, \ldots, \S_k$ where $S \in \S_i$ if 
$$
\frac{n}{2^{i}} \le |S| < \frac{n}{2^{i-1}} ,
$$
for $i = 0, \ldots, k$ (by definition, $\S_0$ contains at most one element).
Fix an index $i = 0, \ldots, k$. For each $S \in \S_i$, we modify~(\ref{eq:decomposition1}) by iterating
from $k$ down to $i$, that is, we eliminate $A_j$, $B_j$ from the decomposition for $j=1, \ldots, i-1$,
and replace it by $F_{i-1} \in \F_{i-1}$. Specifically, we have
\begin{equation}
\label{eq:decomposition2}
  S = (\ldots (((F_{i-1} \sqcup A_i) \setminus B_i) \sqcup A_{i+1}) \setminus B_{i+1}) \sqcup \cdots \sqcup A_{j}) \setminus B_{j}) \sqcup \cdots \sqcup A_k) \setminus B_k .
\end{equation}
Indeed, similarly to decomposition~(\ref{eq:decomposition1}), it is easy to verify by induction on the index $j \ge i$ of the sets that 
$F_{j} = (\ldots (((F_{i-1} \sqcup A_i) \setminus B_i) \sqcup A_{i+1}) \setminus B_{i+1}) \sqcup \cdots \sqcup A_j) \setminus B_j$,
and our claim is obtained when $j=k$.
In particular, when $i=1$ we obtain the same decomposition as in~(\ref{eq:decomposition1}), as $F_0 = \emptyset$.

Let us now fix an index $i$, and construct the 
subsets $\F_j^{i}$ of the packings $\F_j$, for each $j=i-1, \ldots, k$, as follows.
For each $S \in \S_i$, we follow its nearest-neighbor chain (where at this time we stop at $F_{i-1}$)
$S = F_k \rightarrow F_{k-1} \rightarrow \cdots \rightarrow F_{j} \rightarrow \cdots \rightarrow F_{i-1}$,
and put in $\F_j^{i}$ the corresponding element $F_j \in \F_j$, $j=i-1, \ldots, k$.
We next show that the size of each of these members $F_j^{i} \in \F_j^{i}$ is at most $O(|S|) = O(n/2^{i-1})$. 
First we show (see Appendix~\ref{app:distance} for the easy proof):

\begin{claim}
\label{clm:distance}
For each of the sets $F_j^{i} \in \F_{j}^{i}$, $j=i-1, \ldots, k$, we have:
$$
|S \triangle F_{j}^{i}| < \frac{n}{2^{j-1}} .
$$
\end{claim}

Combining this with the fact that $|S| < n/2^{i-1}$ by construction (and the obvious relation 
$F_{j}^{i} \subseteq S \cup (S \triangle F_{j}^{i})$), we obtain: 
\begin{corollary}
  \label{cor:size_NN}
  For each of the sets $F_{j}^{i} \in \F_{j}^{i}$, $j=i-1, \ldots, k$, we have:
  $$
  |F_{j}^{i}| = O\left(\frac{n}{2^{i-1}}\right) .
  $$
\end{corollary} 

\noindent{\bf Remark:}
We note that by construction $|F_{j}^{i} \setminus F_{j-1}^{i}|$, $|F_{j-1}^{i} \setminus F_{j}^{i}|$ are bounded by $n/2^{j-1}$, 
and this fact is used later in the entropy method.
Nevertheless, the property that the actual sets $F_{j}^{i}$ have size $O(n/2^{i-1})$ is stronger, and it enables us to prove
a tighter bound on the number of the canonical sets $F_{j}^{i}$ (Theorem~\ref{thm:packing_sensitive}), and thus on the number of
``pair sets'' $F_{j}^{i} \setminus F_{j-1}^{i}$, $F_{j-1}^{i} \setminus F_{j}^{i}$.
This is crucial for our analysis, as it enables us to reduce the bound on the entropy,
from which we will derive the desired discrepancy bound---see below.

\paragraph{Bounding the size of the packing.}

Having a fixed index $i$ as above, we consider all sets $S \in \S_i$ and the corresponding canonical sets $F_j^{i} \in \F_j^{i}$ participating
in decomposition~(\ref{eq:decomposition2}), $j = i-1, \ldots, k$. 
For a fixed index $j = i-1, \ldots, k$, 
Theorem~\ref{thm:packing} implies that the size of $\F_j^{i}$ is $O(2^{jd})$, but
our goal is to show that the actual bound can be made sensitive to the size of the sets in 
$\F_j^{i}$, that is, to $O(n/2^{i-1})$.

\begin{theorem}[Sensitive Packing Lemma]
  \label{thm:packing_sensitive}
  $$
  |\F_{j}^{i}| = O\left(\frac{j^d 2^{jd}} {2^{(d-d_1) (i-1)}} \right) .
  $$
\end{theorem}
\begin{proof}
  We use a variant of the analysis in~\cite[Chapter 5]{Mat-99} and refine it for our scenario.
  Put $\delta := n/2^{j}$.
  Since the following analysis will restrict sets $S \in \S$ to subsets of $X$, we will refer to $|S|$
  more explicitly as $|S \cap X|$, for the sake of presentation.
 
  We form the set system $(X,\D)$, where $\D = \{ S_1 \triangle S_2 \mid S_1, S_2 \in \S\}$.
  As observed in~\cite{Mat-99}, this set system admits an \emph{$\eps$-net} of size $O((1/\eps) \log{(1/\eps)})$,
  with a constant of proportionality depending on $d$ (see our discussion in Section~\ref{sec:intro}). 
  In fact, a random sample of that size with a sufficiently
  large constant of proportionality (that depends on $d$) is an $\eps$-net with probability at least $3/4$, say.
  
  Set $\eps := \delta /n = 1/2^{j}$ and let $N \subseteq X$ be an $\eps$-net for $(X,\D)$; $|N| = O(j 2^{j})$, with a constant
  of proportionality as above. 
  By the $\eps$-net property, whenever the symmetric difference between two sets 
  $S_1, S_2 \in \S$ has at least $\eps n = \delta$ elements, we must have $(S_1 \triangle S_2) \cap N \neq \emptyset$.
  We thus must have $S_1 \cap N \neq S_2 \cap N$ for any pair of such sets.
  This implies that the number of sets in the packing $\F_j^{i}$ does not exceed the number of sets in the set system
  $(X,\S)$ projected onto $N$ (as each set in $\F_j^{i}$ is mapped to a distinct set in this projected set system). 
  Recall that in view of Corollary~\ref{cor:size_NN} we are interested only in those sets whose 
  size is $O(n/2^{i-1})$.

  We now claim that when we project $(X,\S)$ onto $N$, each set $S \in \S$ with $|S \cap X| = O(n/2^{i-1})$
  satisfies $|S \cap N| = O(j 2^{j} / 2^{i-1})$, with probability at least $1/2$, and thus we only need to bound
  the number of sets of this size in the projected set system.
  Indeed, since $N$ is a random sample of size $O(j 2^{j})$ with a sufficiently large constant of proportionality
  (that depends on $d$), it is \emph{also} a relative $(1/2^{j}, 1/4)$-approximation for $(X,\S)$ with probability at least 
  $3/4$ (see our discussion in Section~\ref{sec:intro} and~\cite{HS-11}). 
  In particular, this means that $N$ is both a $(1/2^{j})$-net for $(X,\D)$ and a relative $(1/2^{j}, 1/4)$-approximation for 
  $(X,\S)$ with probability at least $1/2$ (and thus we obtain a single sample with a ``double role'').
  The latter property implies that 
  $$
  \left| \frac{|S \cap N|}{|N|} - \frac{|S \cap X|}{|X|} \right | \le \frac{1}{4} \cdot \frac{|S \cap X|}{|X|} ,
  $$
  if $\frac{|S \cap X|}{|X|} \ge 1/2^{j}$, and
  $$
  \left| \frac{|S \cap N|}{|N|} - \frac{|S \cap X|}{|X|} \right | \le \frac{1}{4} \cdot \frac{1}{2^{j}} ,
  $$
  otherwise.
  Combining the facts that $|S \cap X| = O(n/2^{i-1})$, $j \ge i-1$, and the bound on $|N|$, we obtain that 
  $|S \cap N| = O(j 2^{j} / 2^{i-1})$, as asserted. In particular, this is also the bound on the size of the sets $F_j^{i} \in \F_j^{i}$ 
  restricted to $N$.
  
  Let $\S_{N}$ be the family $\S$ projected onto $N$.
  By definition, the number of sets of $(N, \S_{N})$ of size $1 \le k \le |N|$ is $O(|N|^{d_1} k^{d-d_1})$, 
  and hence the number of its subsets of size $O(j 2^{j} / 2^{i-1})$ is at most 
  $$
  O\left(|N|^{d_1} \left(j 2^{j} / 2^{i-1} \right)^{d - d_1} \right) = O\left(\frac{j^{d} 2^{jd}} {2^{(d-d_1)(i-1)} } \right) ,
  $$
  from which we obtain the bound on $|\F_{j}^{i}|$.
\end{proof} 

\paragraph{Applying the entropy method.}
Let us fix an index $i$ for the family $\S_i$ under consideration.
Returning to decomposition~(\ref{eq:decomposition2}), we denote, with a slight abuse of notation,
the canonical sets $A_j$, $B_j$ by $A_j^{i}$, $B_j^{i}$, respectively.
By construction, $A_j^{i} = F_j^{i} \setminus F_{j-1}^{i}$, $B_j^{i} = F_{j-1}^{i} \setminus F_{j}^{i}$.
Let $\M_{j}^{i}$ be the collection of these sets $A_j^{i}$, $B_j^{i}$ (or $F_{j}^{i}$ if $j = i-1$).
We now set a parameter $\Delta_{j}^{i}$ for the discrepancy bound (with respect to partial coloring) 
for each of the canonical sets in $\M_{j}^{i}$, where all sets in this subcollection are assigned the same 
discrepancy $\Delta_{j}^{i}$.
We then use the entropy method on this new set system in order to obtain a partial coloring $\chi$ achieving the pre-assigned discrepancy bounds. 
Having these bounds at hand, we can then conclude that the discrepancy of any $S \in \S_i$ with respect to $\chi$ is at most
$2 \sum_{j=i-1}^{k} \Delta^{i}_j$ (using standard arguments, see, e.g.,~\cite{Mat-99}), and this will yield the desired size-sensitive bound
for $\chi(\S)$.

In order to apply the entropy method as presented in Proposition~\ref{prop:constructive_disc} we need to have,
for each $j = i-1, \ldots k$,
(i) a bound on $|\M_{j}^{i}|$,  
(ii) a bound on the size of the canonical sets in $\M_{j}^{i}$, 
and (iii) an appropriate choice for the parameter $\Delta^{i}_j$, uniformly assigned
to all these sets.
We set each of the bounds in (i)--(iii) for a fixed index $i$, and then sum them up over all 
$i=1, \ldots, k$.

For the bound in (i), we observe that each canonical set $A_j^{i}$ (resp., $B_j^{i}$) corresponds to a pair
$(F_j^{i}, F_{j-1}^{i})$, but each of these pairs can uniquely be charged to $F_j^{i}$, as $F_{j-1}^{i}$ is the
corresponding nearest neighbor in the packing $\F_{j-1}^{i}$. Thus the number of these canonical sets is $|\F_{j}^{i}|$,
for $j = i, \ldots, k$. For $j=i-1$, the bound is trivially $|\F_{j}^{i}|$.  
Thus by applying the Sensitive Packing Lemma (Theorem~\ref{thm:packing_sensitive}) we conclude that
the overall number of canonical sets is at most $ C \cdot \frac{j^d 2^{jd}} {2^{(d  - d_1)(i-1)}}$, for $j=i-1, \ldots, k$, 
where $C > 0$ is an appropriate constant whose choice depends on $d$.
By construction, the size of the sets $A_j^{i}$, $B_j^{i}$ is at most $s_j = n/2^{j-1}$ (recall that $|F_j^{i} \triangle F_{j-1}^{i}| \le n/2^{j-1}$)
and $|F_{i-1}^{i}| = O(n/2^{i-1})$ by Corollary~\ref{cor:size_NN}, 
whence we get the bound for (ii). For the choice in (iii) we set:
\begin{equation}
  \label{eq:Delta_choice}
  \Delta_{j}^{i} := A \frac{1}{(1 + |j - j_0|)^2} \left(\frac{n^{1/2 - 1/(2d)}}{2^{(1/2 - d_1/(2d)) \cdot (i-1)}} \right) \log^{1/2 + 1/2d}{n} ,
\end{equation}
where $j_0 := (1/d) \log{n} + (1 - d_1/d)(i-1) - (1 + 1/d)\log\log{n} - B$, for an appropriate constant $B > 5 + \log{C}$,
and for a sufficiently large constant of proportionality $A > 0$, whose choice depends on $B$, and will be 
determined shortly (note that all the three constants $A$, $B$, and $C$ depend on $d$). 

We provide a brief justification for our choice of $j_0$ and the appearance of the coefficient $\frac{1}{(1 + |j - j_0|)^2}$.
For the first, our goal is to bound the entropy function, bounded by the sum~(\ref{eq:entropy}) appearing in 
Proposition~\ref{prop:disc_canonical_set}, and at $j=j_0$ the corresponding summands achieve their maximum value 
(which is some linear function of $n/\log{n}$ with an appropriate constant of proportionality). 
Then when $j \ge j_0$ the exponents ``take over'' this summation, in which case it decreases superexponentially, and 
when $j < j_0$ the terms $C \cdot \frac{j^d 2^{jd}} {2^{(d - d_1)(i-1)}}$ representing the packings take over this summation and decrease geometrically.
This eventually leads to the bound stated in Proposition~\ref{prop:disc_canonical_set}.

The coefficient $\frac{1}{(1 + |j - j_0|)^2}$ in~(\ref{eq:Delta_choice}) 
guarantees that the sum $\sum_{j=i-1}^k \Delta_{j}^{i}$ (corresponding to the asymptotic bound 
on the discrepancy of any $S \in \S_i$) converges to $O\left( n^{1/2 - 1/2d}/2^{(1/2 - d_1/(2d)) \cdot (i-1)} \log^{1/2 + 1/2d}{n}\right)$.
In particular, it does not contain an extra logarithmic factor over the initial bound of $\Delta_{j}^{i}$---see below. 
Similar ideas have been used in~\cite{Mat-99}.
We defer the remaining technical details to Appendix~\ref{app:disc_canonical_set}, where we present the proof of 
Proposition~\ref{prop:disc_canonical_set}.

\begin{proposition}
  \label{prop:disc_canonical_set}
  The choice in~(\ref{eq:Delta_choice}), for $A > 0$ sufficiently large (whose choice depends on $C$
  and thus on $d$), satisfies
  \begin{equation}
    \label{eq:entropy}
    \sum_{i=1}^k \sum_{j=i-1}^{k} C \cdot \frac{j^d 2^{jd}} {2^{(d - d_1)(i-1)}} 
    \exp{\left( -\frac{{(\Delta_{j}^{i})}^2}{16 s_j}\right)} \le 
    \frac{n}{16} .
  \end{equation}
\end{proposition}

\noindent{\bf Remark:}
Currently, our analysis is slightly suboptimal, as our bound in~(\ref{eq:Delta_choice}) contains 
an extra fractional logarithmic power, which results from the following reasons:
(i) We may overcount in our bound on the entropy function the same 
set $A_j^{i}$ (or $B_j^{i}$) when we iterate over $i=1, \ldots, k$; recall that for each $S \in \S_i$ we follow its nearest-neighbor 
chain and put in each $\F_j^{i}$ the corresponding element from $\F_j$, thus an element from $\F_j$ may appear in multiple layers $i$.
(ii) In the Sensitive Packing Lemma (Theorem~\ref{thm:packing_sensitive}) the size of the random sample $N$ is $O((n/\delta) \log{(n/\delta)})$, 
whereas the analysis of the original Packing Lemma (Theorem~\ref{thm:packing}), resulting in the bound $O((n/\delta)^d)$, 
considers a sample of only $O(n/\delta)$ elements. 
Nevertheless, due to the fact that our sample is also a relative approximation (to exploit the property that the packing contains 
sets of size at most $O(n/2^{i-1})$), we had to use a slightly larger sample. 
It is an interesting open problem whether in the case $d_1 = 1$ (where the set system is ``well-behaved'') the improved bound on the size of 
relative approximations, recently shown by the author~\cite{Ezra-13}, can be integrated into the analysis of the original Packing Lemma.

%

\paragraph{Wrapping up.}

Incorporating Propositions~\ref{prop:constructive_disc} and~\ref{prop:disc_canonical_set}, we obtain that there exists a partial
coloring $\chi$, computable in expected polynomial time, which colors at least $n/2$ points of $X$, such that
$\chi(M_j^{i}) \le \Delta_{j}^{i} + 1/n^c$, for each $M_{j}^{i} \in \M_{j}^{i}$, 
where $c > 0$ is an arbitrarily large constant. 
But then our choice in~(\ref{eq:Delta_choice}) and our earlier discussion imply that, for each $S \in \S_i$, 
$$
\chi(S) \le 2\sum_{j=i-1}^{k} \Delta^{i}_j =  O\left( \frac{n^{1/2- 1/(2d)}}{2^{(i-1) \cdot (1/2 - d_1/(2d))}} \log^{1/2 + 1/2d}{n} \right),
$$
since the sum $\sum_{j=i-1}^{k}\frac{1}{(1 + |j - j_0|)^2}$ converges. Due to the fact that  $n/2^{i} \le |S| < n/2^{i-1}$ (by definition),
the latter term is:
$$
O\left( |S|^{1/2 - d_1/(2d)} n^{(d_1-1)/(2d)} \log^{1/2 + 1/(2d)}{n} \right) .
$$

Applying the partial coloring procedure (Proposition~\ref{prop:constructive_disc}) for at most $\log{n}$ iterations 
(which yields a full coloring for $X$), we obtain that for each $S \in \S_i$,  
$\chi(S) = O\left( |S|^{1/2 - d_1/(2d)} n^{(d_1-1)/(2d)} \log^{3/2 + 1/2d}{n} \right)$.
We have just shown:

\begin{theorem}
  \label{thm:disc_bound}
  Let $(X,\S)$ be a (finite) set system of primal shatter dimension $d$ with the additional property that 
  in any set system restricted to $X' \subseteq X$, the number of sets of size $k \le |X'|$ is 
  $O(|X'|^{d_1} k^{d - d_1})$, $1 \le d_1 \le d$. Then 
  $$
  \disc(\S) =  O\left( |S|^{1/2 - d_1/(2d)} n^{(d_1 - 1)/(2d)} \log^{3/2 + 1/2d}{n} \right),
  $$ 
  where the constant of proportionality depends on $d$.
  In particular, when $(X,\S)$ is well-behaved (that is, $d_1=1$), the bound becomes
  $O\left( |S|^{1/2 - 1/(2d)} \log^{3/2 + 1/2d}{n} \right)$.
\end{theorem}

\noindent{\bf Remark:}
We note that although the number of uncolored points in $X$ decreases by at least a half after applying 
Proposition~\ref{prop:constructive_disc}, it does not necessarily guarantee that the size of a set $S \in \S$ 
decreases by the same factor, and so we can bound it from above only by $n/2^{i}$ at each round. 
Moreover, it may happen that a set $S \in \S_i$ from the previous round now lies in a different class at the current partition.
Thus at each round we need to resume the process from scratch due to which we obtain an extra logarithmic factor as shown in the bound above. 

\paragraph{Algorithmic aspects.}


In order to apply the randomized algorithm of Lovett and Meka~\cite{LM-12}, we first need to construct, for each $i=1, \ldots, k$, 
the canonical sets in $\F_{j}^{i}$. 

In order to do so for a fixed $i$, we need to construct a $\delta$-packing, for $\delta = n/2^{j}$, $j = i-1, \ldots, k$, such that 
the size of each set in the packing does not exceed $K n/2^{i-1}$, for an appropriate constant $K > 0$.
We thus iterate over $j = i-1, \ldots, k$, and form a $\delta$-packing $\F_j^{i}$ as above in a brute force 
manner by initially picking an arbitrary set $F \in \S$, whose size is at most $K n/2^{i-1}$, to be included into
$\F_{j}^{i}$, and then keep collecting sets $F' \in \S$ into $\F_j^{i}$ if (i) $|F'| \le K n/2^{i-1}$ and 
(ii) the distance between $F'$ and each of the elements currently in $\F_j^{i}$ is at least $\delta$. 
We stop as soon as there are no leftover sets $F'$ of the above kind. The set just created is 
inclusion-maximal and thus according to the Sensitive Packing Lemma (Theorem~\ref{thm:packing_sensitive}) its size is only 
$O\left(\frac{j^d 2^{jd}} {2^{(d - d_1)(i-1)}} \right)$.
It is easy to verify that the construction of each $\delta$-packing $\F_j^{i}$ can be performed in polynomial time due to the fact 
that the number of sets in $\S$ is only $O(n^{d})$. 
Omitting any further details we conclude:

\begin{corollary}
  \label{cor:disc}
  A coloring $\chi$ achieving the discrepancy bound in Theorem~\ref{thm:disc_bound} can be computed in expected polynomial time.
\end{corollary}

\paragraph{The case of points and halfspaces in $d$ dimensions.}

When $(X,\S)$ is a set system of points and halfspaces in $d$-space, it is known that the number of halfspaces containing
at most $k$ points of $\S$ is $O(n^{\lfloor{d/2}\rfloor} k^{\lceil{d/2}\rceil})$ (see, e.g.,~\cite{CS-89}).
Thus from Theorem~\ref{thm:disc_bound} and Corollary~\ref{cor:disc} we conclude:
\begin{theorem}
  \label{thm:halfspaces}
  Let $(X,\S)$ be a set system of points and halfspaces in $d$-space.
  Then 
  $$
  \disc(\S) =  O\left( |S|^{1/4} n^{1/4 - 1/(2d)} \log^{3/2 + 1/2d}{n} \right),
  $$ 
  if $d$ is even, and 
  $$
  \disc(\S) =  O\left( |S|^{1/4 + 1/(4d)} n^{1/4 - 3/(4d)} \log^{3/2 + 1/2d}{n} \right) ,
  $$
  if $d$ is odd.
  In particular, when $d=2,3$, these bounds become
  $O\left( |S|^{1/4} \log^{7/4}{n} \right)$, $O\left( |S|^{1/3} \log^{5/3}{n} \right)$,
  respectively.
  The constant of proportionality in these bounds depends on $d$, 
  and the corresponding coloring $\chi$ can be computed in expected polynomial time. 
\end{theorem}

\paragraph{Concluding remarks and further research.}

We note that whereas our construction is a variant of that of Matou\v{s}ek~\cite{Mat-95}, a key ingredient in our analysis is the 
Sensitive Packing Lemma (Theorem~\ref{thm:packing_sensitive}), where we restrict each set in a $\delta$-packing to have a bounded size. 
As our analysis shows, the number of such sets (when $i$ is not too small) is considerably smaller than the bound 
$O((n/\delta)^{d})$ derived in the (original) Packing Lemma (Theorem~\ref{thm:packing}). 
This bound is eventually integrated into the entropy method applied in
Proposition~\ref{prop:disc_canonical_set} (in addition to our decomposition~(\ref{eq:decomposition2})), 
from which we eventually obtain the discrepancy bound in~(\ref{eq:Delta_choice}).

This study raises several open problems, some of which are under on-going research.
First, we are now in the process of deriving improved bounds for relative $(\eps,\delta)$-approximations
using our discrepancy bounds in Theorem~\ref{thm:disc_bound}. A major technical difficulty is the fact that,
unlike the colorings given in~\cite{HS-11, SZ-12}, our coloring is not necessarily balanced and this requires
an extra care when one applies the ``halving technique'' in~\cite{HS-11} (see also~\cite{Chaz-01}). 
In Appendix~\ref{app:relative_approx} we present our considerations for deriving a relative $(\eps,\delta)$-approximation 
for well-behaved set systems (that is, $d_1=1$). We plan to finalize these details and extend them to the more general case 
in the full version of this paper.

Another question, related to the remarks following Proposition~\ref{prop:disc_canonical_set} and Theorem~\ref{thm:disc_bound} 
is whether the logarithmic factor in our discrepancy bound can be removed, in which case it becomes optimal.
Initially, it will be interesting to revisit the case of points in the plane and halfplanes, studied by Har-Peled and Sharir~\cite{HS-11},
and reduce the logarithmic factor to $\log{n}$ (to be matched with their bound). As stated earlier in the paper, for this particular 
case (and to well-behaved set systems in general) we hope to be able to integrate the recent improved bounds for relative 
$(\eps,\delta)$-approximations~\cite{Ezra-13} with the analysis in this paper.

Last but not least, we are interested in the implications of our discrepancy bound to \emph{approximate range search}.
In particular, for points and halfspaces in ${\reals}^d$, can we improve the query time in the approximate range-counting 
machinery of Aronov and Sharir~\cite{AS-10}?



\paragraph*{Acknowledgments.}
The author wishes to thank Boris Aronov and Micha Sharir for helpful discussions and for their
tremendous help in writing this paper.

\appendix

\section{The Construction}
\label{app:construction}

\subsection{Proof of Claim~\ref{clm:distance}:}
\label{app:distance}
By construction we have $|F_j^{i} \triangle F_{j-1}^{i}| \le n/2^{j-1}$, $F_j^{i} \in \F_j^{i}$, and $F_{j-1}^{i} \in \F_{j-1}^{i}$,
for each $j=i, \ldots, k$.
We now apply the triangle inequality on the symmetric 
difference\footnote{This follows from the property that for each triple of sets $X$, $Y$, $Z$, 
  we have $X \triangle Z \subseteq (X \triangle Y) \cup (Y \triangle Z)$.} 
and obtain:
$$
|S \triangle F_{j}^{i}| \le  
|S \triangle F_{k-1}^{i}| + |F_{k-1}^{i} \triangle F_{k-2}^{i}| + \cdots + |F_{j+1}^{i} - F_{j}^{i}|
\le
\frac{n}{2^{k-1}} + \frac{n}{2^{k-2}} + \cdots + \frac{n}{2^{j}}
< \frac{n}{2^{j-1}} ,
$$
as asserted.

\subsection{Proof of Proposition~\ref{prop:disc_canonical_set}:}
\label{app:disc_canonical_set}
  
We first note that at $j_0$ the above exponent becomes a constant, whereas the size of the packing 
becomes roughly $n/\log{n}$ (for a fixed index $i$). Indeed, applying our choice in~(\ref{eq:Delta_choice}), we have
$$
\exp{\left( -\frac{{(\Delta_{j}^{i})}^2}{16 s_j}\right)} = 
\exp{\left( -\frac{A^2 \cdot 2^{j-1}\log^{1 + 1/d}{n}}{16(1 + |j - j_0|)^4 n^{1/d} 2^{(1 - d_1/d) \cdot (i-1)}} \right)} ,
$$
which is $\exp{\left(-\frac{A^2}{16 \cdot 2^{B+1}}\right)}$ at
$j =j_0 = (1/d) \log{n} + (1 - d_1/d)(i-1) - (1 + 1/d)\log\log{n} - B$.
Concerning the bound on the packing size, $C \cdot \frac{j^{d} 2^{jd}} {2^{(d - d_1)(i-1)} }$, 
since $j$ can always be bounded by $k=\log{n}$, at $j=j_0$ we obtain:
$$
C \cdot \frac{n 2^{(d - d_1)(i-1)} \log^d{n} } {2^{(d - d_1)(i-1)} 2^{dB} \log^{d+1}{n}} =
C \cdot \frac{n}{2^{dB}\log{n}} .
$$

We now fix an index $i$, split the summation into the two parts $j \ge j_0$ and $i-1 \le j < j_0$, and then 
bound each part in turn. In the first part, the exponent will ``take over'' the summation in the sense that it
decreases superexponentially, making the other factors insignificant, and in the second part, 
the packing size will decrease geometrically. Thus the
``peak'' of this summation is obtained at $j = j_0$, and is decreasing as we go beyond or below $j$.

For the first part, put $j := j_0 + l$, for an integer $l \ge 0$, and then
$$
\sum_{i=1}^k \sum_{j=j_0}^{k} 
C \cdot \frac{j^d 2^{jd}} {2^{(d - d_1) (i-1)}} \exp{\left( -\frac{{(\Delta_{j}^{i})}^2}{16 s_j}\right)} \le
\sum_{i=1}^k \sum_{l=0}^{k - j_0} 
C \cdot \frac{n 2^{ld} } { 2^{dB} \log{n}} 
\exp{\left( - \frac{A^2}{16} \cdot \frac{ 2^{l-(B+1)} }{(1 + l)^4} \right) } 
$$
$$
\le
C \cdot n 2^{-dB} \sum_{l=0}^{k-j_0} 
2^{ld} \exp{\left( - \frac{A^2}{16} \cdot \frac{ 2^{l-(B+1)}}{(1 + l)^4} \right) } ,
$$ 
where the logarithmic factor in the packing size is now eliminated due to the summation over $i$.
The exponents in the above sum decrease superexponentially.
Choosing $A$ sufficiently large (say, $A > 2^{6 + (B+1) + \log{d}}$) and having $B > 5 + \log{C}$ as above, we can 
guarantee that the latter sum is strictly smaller than $n/32$.

When $j < j_0$, put $j := j_0 - l$, $l > 0$ as above. We now obtain, by just bounding the exponent from above by $1$,
and using similar considerations as above:
$$
\sum_{i=1}^k \sum_{j=i-1}^{j_0 - 1} 
C \cdot \frac{j^d 2^{jd}} {2^{(d - d_1) (i-1)}} \exp{\left( -\frac{{\Delta_{j}^{i}}^2}{16 s_j}\right)} \le
\sum_{l=1}^{j_0 - (i-1)} C \cdot \frac{n}{2^{d(l + B)}} .
$$
Once again, our choice for $B$ guarantees that the above (geometrically decreasing) sum is strictly smaller than $n/32$. 
Thus the entire summation is bounded by $n/16$, as asserted.

\section{Applications to Relative $(\eps,\delta)$-Approximations}
\label{app:relative_approx}

In this section, we present some of our initial ideas in constructing relative $(\eps,\delta)$-approximations after having
a size-sensitive discrepancy bound at hand.
Our construction is a variant of the ``halving technique'' appeared in~\cite{HS-11} (based on a technique in~\cite{Chaz-01}), where the 
main difference between this construction to ours is the fact that the coloring $\chi$ that we produce is not necessarily balanced.
At this version we present our considerations only for the case $d_1 = 1$ (well-behaved set systems).
Following the arguments in~\cite{HS-11}, it is sufficient to construct a $(\nu,\alpha)$-sample for $(X,\S)$ as such a sample
is equivalent to a relative $(\eps,\delta)$-approximation (where $\nu$, $\eps$ and $\alpha$, $\delta$ are within some constant factor 
from each other---see Section~\ref{sec:intro}).

Our construction proceeds over iterations, where we repeatedly ``halve'' $X$ until we obtain a subset of an appropriate size, 
which we argue to comprise the resulting $(\nu,\alpha)$-sample. Put $X_0 := X$. Then, at each iteration $i \ge 1$, 
we let $X_i$, $X_i'$ be the two corresponding portions of $X_{i-1}$, where the points in $X_i$ are, say, colored $+1$ and the points in 
$X_i'$ are colored $-1$. Assume, w.l.o.g., $|X_i| \ge |X_i'|$. We now keep $X_i$, remove $X_i'$ and continue in this ``halving'' process.
Put $n_i := |X_i|$. On an even split we have $n_i = n/2^{i}$ (recall that we assume $n$ is an integer power of $2$), nevertheless, with
our coloring $\chi$ producing the discrepancy bound in Theorem~\ref{thm:disc_bound}, $n_i$ may be slightly larger. We bound its size
as follows.

We first can assume, w.l.o.g., $X = X_0$ is part of $\S$, and thus, at each iteration $i$, $X_{i-1}$ is part of the collection $\S$ projected
onto $X_{i-1}$, $i \ge 1$.
Applying Theorem~\ref{thm:disc_bound} at iteration $i$ we obtain
$$
\left| |S \cap X_i| - |S \cap X_i'| \right| \le K \cdot |S \cap X_{i-1}|^{1/2 - 1/(2d)} \log^{3/2 + 1/(2d)}{|X_{i-1}|} ,
$$
for an appropriate constant $K > 0$.
Letting $S = X_{i-1}$, we obtain:
$$
\left| |X_i| - (|X_{i-1}| - |X_i|) \right| \le K \cdot |X_{i-1}|^{1/2 - 1/(2d)} \log^{3/2 + 1/(2d)}{|X_{i-1}|}  ,
$$
from which we obtain
$$
|X_i| \le \frac{|X_{i-1}|}{2} \left(1 + \frac{K\log^{3/2 + 1/(2d)}{|X_{i-1}|}}{|X_{i-1}|^{1/2 + 1/(2d)}} \right).
$$
We thus write the bound on $|X_i|$ as
\begin{equation}
  \label{eq:X_i}
  |X_i| = \frac{|X_{i-1}|}{2} \left(1 + \delta_{i-1} \right) ,
\end{equation}
where $0 \le \delta_{i-1} \le \frac{K \log^{3/2 + 1/(2d)}{|X_{i-1}|}}{|X_{i-1}|^{1/2 + 1/(2d)}}$.
Applying~(\ref{eq:X_i}) recursively on $i$, we obtain:
$$
|X_i| = \frac{|X_0|}{2^i} \prod_{j=0}^{i-1} (1 + \delta_{j}) \le \frac{|X_0|}{2^i} \exp{\left\{ \sum_{j=0}^{i-1} \delta_j \right\} }.
$$
From the bound on $\delta_j$ we have (recall that $n_i = |X_i|$):
$$
n_i \le \frac{n}{2^i} \exp{ \left\{ K \log^{3/2 + 1/(2d)}{n} \sum_{j=0}^{i-1} \left(\frac{2^{i}}{n}\right)^{1/2 + 1/(2d)} \right \} } ,
$$
and the exponent in the latter term is $O(1)$ when 
$i \le i^{*} = \log{n} - \frac{3/2 + 1/(2d)}{1/2 + 1/(2d)} \cdot \log\log{n} - \log^{\frac{1}{1/2 + 1/(2d)}}{K}$.
We thus stop the process at iteration $i^{*}$ (or earlier---see below), from which we obtain a lower bound of 
$\Omega(\log^{\frac{3/2 + 1/(2d)}{1/2 + 1/(2d)}}{n})$ on $n_{i-1}$ (with a constant of proportionality depending on $K$).
We next proceed with the presentation of the ``halving'' process in order to obtain a relative error of $\alpha$, and then integrate 
the resulting bound with the one above.

Our next goal is to bound, for each $S \in \S$, the difference 
$\left| \frac{|S \cap X_{i-1}|}{|X_{i-1}|} - \frac{|S \cap X_{i}|}{|X_{i}|} \right|$, 
which we also denote by $|\overline{X_{i-1}}(S) - \overline{X_{i}}(S)|$.
Since $|X_i| = \frac{|X_{i-1}|(1 + \delta_{i-1})}{2}$ and $X_{i-1} = X_i \sqcup X_i'$, this difference is
$$
\left|\frac{|S \cap X_i| + |S \cap X_i'|}{|X_{i-1}|} - \frac{2|S \cap X_i|}{|X_{i-1}|(1+\delta_{i-1})} \right| =
\left| \frac{|S \cap X_{i}'|}{|X_{i-1}|} - \frac{|S \cap X_i|(1-\delta_{i-1})}{|X_{i-1}|(1+\delta_{i-1})} \right| .
$$
By adding and subtracting $\frac{|S \cap X_i|}{|X_{i-1}|}$, we obtain that the latter term is:
$$
\left| \frac{|S \cap X_{i}'| - |S \cap X_i|}{|X_{i-1}|} + \frac{2\delta_{i-1}}{1 + \delta_{i-1}} \cdot \frac{|S \cap X_i|}{|X_{i-1}|} \right| 
$$
$$
= 
\left| \frac{|S \cap X_{i}'| - |S \cap X_i|}{|X_{i-1}|} + \delta_{i-1} \cdot \frac{|S \cap X_i|}{|X_i|} \right| \le
\left| \frac{|S \cap X_{i}'| - |S \cap X_i|}{|X_{i-1}|} \right| +  \delta_{i-1} \cdot \overline{X_{i}}(S) .
$$
since $|X_{i-1}| = 2\frac{|X_i|}{1 + \delta_{i-1}}$.
By our discrepancy bound in Theorem~\ref{thm:disc_bound} we have 
$$
|S \cap X_{i}'| - |S \cap X_i| \le  K \cdot |S \cap X_{i-1}|^{1/2 - 1/(2d)} \log^{3/2 + 1/2d}{|X_{i-1}|} ,
$$
for an appropriate constant $K > 0$. We thus obtain:
$$
|\overline{X_{i-1}}(S) - \overline{X_{i}}(S)| \le \frac{K \cdot |S \cap X_{i-1}|^{1/2 - 1/(2d)} \log^{3/2 + 1/2d}{|X_{i-1}|}}{|X_{i-1}|} + 
 \delta_{i-1} \cdot \overline{X_{i}}(S) .
$$
The latter term in the above sum is obviously bounded by $\delta_{i-1} \cdot (\overline{X_{i}}(S) + \overline{X_{i-1}}(S) + \nu)$.
Concerning the first term, we write it as 
$$
\frac{K {\overline{X_{i-1}}(S)}^{1/2 - 1/(2d)} }{|X_{i-1}|^{1/2 + 1/(2d)}} \log^{3/2 + 1/2d}{|X_{i-1}|} .
$$
and use the observation that $x^{p} < (x+y)/y^{1-p}$, for $x \ge 0$, $y > 0$ and $0 < p < 1$ (stated in~\cite{HS-11})
in order to bound the latter term by
$$
\frac{K \cdot \log^{3/2 + 1/2d}{n_{i-1}} }{n_{i-1}^{1/2 + 1/(2d)}} \cdot \frac{\overline{X_{i-1}}(S) + \nu}{\nu^{1/2 + 1/(2d)}} 
\le
\frac{K \cdot \log^{3/2 + 1/2d}{n_{i-1}} }{n_{i-1}^{1/2 + 1/(2d)}} \cdot \frac{\overline{X_{i-1}}(S) + \overline{X_{i-1}}(S) + \nu}{\nu^{1/2 + 1/(2d)}} .
$$

This implies that 
$$
d_{\nu}(\overline{X_{i-1}}(S), \overline{X_i}(S)) = 
\frac{|\overline{X_i}(S) - \overline{X_{i-1}}(S)|}{\overline{X_i}(S) + \overline{X_{i-1}}(S) + \nu} 
$$
$$
\le
\frac{K \cdot \log^{3/2 + 1/2d}{n_{i-1}}}{(\nu n_{i-1})^{1/2 + 1/(2d)}} + \delta_{i-1} \le
\frac{K \cdot \log^{3/2 + 1/2d}{n_{i-1}}}{(n_{i-1})^{1/2 + 1/(2d)}} \left( \frac{1}{\nu^{1/2 + 1/(2d)}} + 1 \right)  ,
$$
due to the bound on $\delta_{i-1}$.

Since $d_{nu}(\cdot, \cdot)$ satisfies the triangle inequality (see~\cite{LLS-01}), we obtain:
$$
d_{\nu}(\overline{X_0}(S), \overline{X_i}(S)) \le \sum_{j=1}^{i} d_{\nu}(\overline{X_{j-1}}(S), \overline{X_j}(S)) 
$$
$$
\le
K \cdot \left( \frac{1}{\nu^{1/2 + 1/(2d)}} + 1 \right)  \sum_{j=1}^{i} \frac{ \log^{3/2 + 1/(2d)}{n_{j-1}} }{ n_{j-1}^{1/2 + 1/(2d)}} 
= O\left( \frac{\log^{3/2 + 1/(2d)}{n_{i-1}}}{(\nu n_{i-1})^{1/2 + 1/(2d)}} \right) \le \alpha ,
$$
for $n_{i-1} = \Omega\left(\frac{\log^{(3 + 1/d)/(1+1/d)}{\frac{1}{\nu \alpha}} }{\nu \alpha^{2/(1+1/d)}} \right)$.
We thus stop at that iteration $i$ for which the set $X_{i-1}$ is the smallest that still satisfies 
this lower bound.

Combining these considerations with the fact that we stop the process no later than iteration $i^{*}$, we obtain
$$
n_{i-1} = \max \left \{ \Omega\left(\log^{\frac{3/2 + 1/(2d)}{1/2 + 1/(2d)}} {n}\right), 
\Omega\left(\frac{\log^{\frac{3 + 1/d}{1+1/d}}{\frac{1}{\nu \alpha}} }{\nu \alpha^{2/(1+1/d)}} \right) \right \} .
$$


\begin{thebibliography}{99}

\balance










\bibitem{Alexander-90}
R.~Alexander. 
\newblock Geometric methods in the theory of uniform distribution. 
\newblock \emph{Combinatorica}, 10(2):115--136, 1990.

\bibitem{AS-00}
N.~Alon and J.~H. Spencer.
\newblock {\em The Probabilistic Method}.
\newblock $2nd$ Edition, Wiley-Interscience, New York, USA, 2000.

\bibitem{AS-10}
B.~Aronov and M.~Sharir.
\newblock Approximate halfspace range counting.
\newblock \emph{SIAM J. Comput.}, 39(7), 2704--2725, 2010. 

\bibitem{Bansal-10}
N.~Bansal.
\newblock Constructive algorithms for discrepancy minimization.
\newblock In \emph{Proc. 51st Annu. IEEE Symp. Found. Comput. Sci.}, pp. 3--10, 2010.

\bibitem{Beck-81}
J.~Beck. 
\newblock Roth’s estimate on the discrepancy of integer sequences is nearly sharp. 
\newblock \emph{Combinatorica}, 1(4):319--325, 1981.

\bibitem{BC-87}
J.~Beck and W.~Chen. 
\newblock \emph{Irregularities of Distribution}. 
\newblock Cambridge University Press, Cambridge, 1987.










\bibitem{Chazelle-92}
B.~Chazelle. 
\newblock A note on Haussler’s packing lemma. 
\newblock Unpublished manuscript, Princeton, 1992.

\bibitem{Chaz-01}
B.~Chazelle.
\newblock \emph{The Discrepancy Method: Randomness and Complexity}.
\newblock Cambridge University Press, 2000; paperback, 2001

\bibitem{Chaz-04}
B. Chazelle.
\newblock The discrepancy method in computational geometry, chapter 44, in
\emph{Handbook of Discrete and Computational Geometry}, 2nd Edition, J.E. Goodman and J.
O'Rourke, Eds., 
\newblock CRC Press, Boca Raton, 2004, pages 983--996.

\bibitem{CW-89}
B.~Chazelle and E.~Welzl. 
\newblock Quasi-optimal range searching in spaces of finite VC-dimension. 
\newblock \emph{Discrete Comput. Geom.}, 4:467--489, 1989.

\bibitem{Ezra-13}
E.~Ezra.
\newblock Small-size relative $(p,\eps)$-approximations for well-behaved range spaces.
\newblock In \emph{Proc. 29th Sympos. Comput. Geom.}, 2013, to appear.





\bibitem{CS-89}
K.~L. Clarkson and P.~W. Shor.
Applications of random sampling in computational geometry, II.
\emph{Discrete Comput. Geom.}, 4(1989):387--421.






















\bibitem{Har-Peled-11}
S.~Har-Peled.
\emph{Geometric Approximation Algorithms},
Mathematical Surveys and Monographs, Vol. 173, 2011.

\bibitem{HS-11}
S.~Har-Peled and M.~Sharir,
Relative $(p,\eps)$-approximations in geometry,
\emph{Discrete Comput. Geom.}, 45(3):462--496 2011.


\bibitem{Haussler-92}
D.~Haussler.
\newblock Decision theoretic generalizations of the PAC model for neural net and other learning applications.
\newblock In \emph{Information and Computation}, 100(1):78--150, 1992.

\bibitem{Haussler-95}
D.~Haussler. 
\newblock Sphere packing numbers for subsets of the Boolean $n$-cube with bounded Vapnik-Chervonenkis dimension. 
\newblock \emph{J. Combinatorial Theory Ser. A}, 69:217--232, 1995. 

\bibitem{HW87}
D.~Haussler and E.~Welzl.
$\varepsilon$-nets and simplex range queries.
\emph{Discrete Comput. Geom.}, 2(1987):127--151.



\bibitem{LLS-01}
Y.~Li, P.~M. Long, and A.~Srinivasan.
\newblock Improved bounds on the sample complexity of learning.
\newblock {\it J. Comput. Sys. Sci.}, 62(3):516--527, 2001.

\bibitem{LM-12}
S.~Lovett and R.~Meka.
\newblock Constructive discrepancy minimization by walking on the edges. 
\newblock In {\em Proc. 53th Annu. IEEE Symp. Found. Comput. Sci.}, 61--67, 2012.










\bibitem{Mat-95}
J.~Matou{\v s}ek.
\newblock Tight upper bounds for the discrepancy of halfspaces.
\newblock \emph{Discrete Comput. Geom.}, 13:593--601, 1995.

\bibitem{Mat-98}
J.~Matou{\v s}ek.
\newblock An $L_{p}$ version of the Beck-Fiala conjecture.
\newblock \emph{Eur J. Combinatorics}, 19(2):175--182, 1998.

\bibitem{Mat-99}
J.~Matou{\v s}ek.
\newblock \emph{Geometric Discrepancy}, Algorithms and Combinatorics, Vol. 18,
\newblock Springer Verlag, Heidelberg, 1999



\bibitem{MS-96}
J.~Matou{\v s}ek and J.~Spencer.
\newblock Discrepancy in arithmetic progressions.
\newblock \emph{J. Amer. Math Soc.}, 9:195--204, 1996.


\bibitem{MWW-93}
  J.~Matou{\v s}ek, E. Welzl, and L. Wernisch. Discrepancy and
  \newblock $\eps$-approximations for bounded VC-dimension. 
  \newblock Combinatorica, 13:455--466, 1993.
  
\bibitem{MN-12}
  S.~Muthukrishnan and A.~Nikolov
  \newblock Optimal private halfspace counting via discrepancy.
  \newblock In \emph{Proc. 44th Annu. ACM Symp. Theory Comput.}. pp. 1285--1292, 2012. 














\bibitem{SZ-12}
M.~Sharir and S.~Zaban,
\newblock Output-sensitive tools for range searching in higher dimensions.
\newblock \emph{unpublished manuscript}.
\texttt{http://www.cs.tau.ac.il/thesis/thesis/zaban.pdf} .

\bibitem{Spencer-85}
J.~Spencer. 
\newblock Six standard deviations suffice. 
\newblock \emph{Trans. Amer. Math. Soc.}, 289(2):679--706, 1985.


\bibitem{VC-71}
V.~Vapnik and A.~Chervonenkis.
\newblock On the uniform convergence of relative frequencies of events to their probabilities.
\newblock \emph{Theory Prob. Appl.}, 16(2):264--280, 1971.

\bibitem{Welzl-88}
E.~Welzl. 
\newblock Partition trees for triangle counting and other range searching problems. 
\newblock In \emph{Proc. 4th Annu. ACM Sympos. Comput. Geom.}, pp 23--33, 1988.

\bibitem{Welzl-92}
E.~Welzl.
\newblock On spanning trees with low crossing numbers.
\newblock In \emph{Data Structures and Efficient Algorithms, Final Report on the DFG Special Joint Initiative}, volume 594 of \emph{Lect. Notes in Comp. Sci.}, Springer-Verlag, Heidelberg, pp. 233--249, 1992.



\end{thebibliography}
\end{document}